\documentclass[pdflatex,sn-mathphys-num,iicol]{sn-jnl}%

\usepackage{graphicx}%
\usepackage{multirow}%
\usepackage{amsmath,amssymb,amsfonts}%
\usepackage{amsthm}%
\usepackage{mathrsfs}%
\usepackage[title]{appendix}%
\usepackage{xcolor}%
\usepackage{textcomp}%
\usepackage{manyfoot}%
\usepackage{booktabs}%
\usepackage{algorithm}%
\usepackage{algorithmicx}%
\usepackage{algpseudocode}%
\usepackage{listings}%
\usepackage{mathabx} %
\usepackage{enumitem}

\usepackage{lineno}

\usepackage{dsfont,bm}
\usepackage{xurl}
\newcommand{\PP}{\mathbb{P}}

\newcommand{\true}{\mathrm{true}}
\newcommand{\fig}{{\bf Fig.}~}

\newcommand{\edit}[1]{#1}

\newtheorem{theorem}{Theorem}%

\begin{document}

\title[Variational Markov chain mixtures]{Variational Markov chain mixtures with automatic component selection}

\author*[1]{\fnm{Christopher E.} \sur{Miles}}\email{chris.miles@uci.edu}

\author*[2]{\fnm{Robert J.} \sur{Webber}}\email{rwebber@ucsd.edu}

\affil[1]{\orgdiv{Department of Mathematics}, \orgname{University of California Irvine}, \orgaddress{\street{419 Rowland Hall}, \city{Irvine}, \postcode{92697}, \state{CA}, \country{United States}}}

\affil[2]{\orgdiv{Department of Mathematics}, \orgname{University of California San Diego}, \orgaddress{\street{9500 Gilman Drive}, \city{La Jolla}, \postcode{92093}, \state{CA}, \country{United States}}}

\abstract{
Markov state modeling has gained popularity in various scientific fields since it reduces complex \edit{time-series} data sets \edit{into} transitions between a few states.
Yet common Markov state modeling frameworks assume a single Markov chain describes the data, so they suffer from an inability to discern heterogeneities.
As an alternative, this paper models time-series data using a \emph{mixture} of Markov chains, and it automatically determines the number of mixture components using the variational expectation-maximization algorithm.
Variational EM simultaneously identifies the number of Markov chains and the dynamics of each chain without expensive model comparisons or posterior sampling.
As a theoretical contribution, this paper identifies the natural limits of Markov state mixture modeling by proving a lower bound on the classification error.
It then presents numerical experiments where variational EM \edit{achieves performance consistent with the theoretically optimal error scaling}.
The experiments are based on synthetic and observational data sets including \texttt{Last.fm} music listening, ultramarathon running, and gene expression.
In each of the three data sets, variational EM leads to 
the identification of meaningful heterogeneities.}


\maketitle

\section{Introduction}

This paper considers data emerging from observations, simulations, and experiments, which \edit{come} in the form of dynamical trajectories.
The trajectories can be modeled by a finite-state Markov chain that jumps between different regions (``states'') of phase space.
This modeling approach is known as \emph{Markov state modeling} \cite{husic2018markov}, and it has become increasingly popular over the past two decades, including applications in chemistry \cite{noe2013projected,roblitz2013fuzzy, kohs2022nonparametric}, biology \cite{chu2017markov,tse2018rare,tan2018exploring,fang2018cell}, and climate science \cite{finkel2023revealing,souza2023transforming,springer2024agnostic}.

\edit{Markov state modeling has roots in classical statistical modeling, including clustering-based discretizations and finite-state Markov chains. 
Yet the data sets are changing. 
Observations, simulations, and experiments are now generating large collections of short, noisy, and heterogeneous trajectories that cannot be described by a single Markov chain. To address the contemporary modeling challenge, this work combines the discretization step in Markov state modeling with a Markov chain mixture approach.
As specific contributions, the paper proposes an efficient modeling algorithm and it derives a non-asymptotic information-theoretic lower bound on the classification error achievable from finite trajectory lengths.}

\subsection{Goals and contributions}

This paper argues that Markov state modeling, which is traditionally applied to one Markov chain at a time, can be applied to learn several Markov chains simultaneously.
Markov state mixture modeling was previously proposed and developed in the works \cite{smyth1997clustering,ramoni2002bayesian,batu2004inferring,melnykov2016clickclust,zhou2021you,das2023clustering, das2023utilizing,de2017use,frydman2005estimation}.
The paper builds on the existing research by asking and answering the following questions:
\begin{itemize}
\item[(a)] What Markov chain mixture modeling approach can be easily adapted to high dimensions?
\item[(b)] What efficient algorithm can optimize the parameters in the model?
\item[(c)] What theoretical accuracy can the Markov mixture model attain?
\end{itemize}

For part (a), the paper applies a simple model \cite{smyth1997clustering} based on a mixture of Markov chains in a finite state space.
Traditionally, the finite states were determined using expert knowledge of the predominant configurations of the system.
However, the modern approach automatically identifies the states with a machine learning method such as spectral clustering.

For part (b), the paper adapts the variational expectation-maximization method \cite[Sec.~10]{bishop2006pattern} to the setting of Markov chain mixture modeling.
Variational EM has a major advantage over the classical EM algorithm, since classical EM requires fitting several models and then performing a post-hoc model comparison to determine the number of mixture components \cite{keribin2000consistent,smyth1997clustering,melnykov2016clickclust}.
In contrast, variational EM determines the number of clusters automatically \cite{corduneanu2001hyperparameters,mackay2001local,rousseau2011asymptotic}.
\edit{
Here the variational EM approach is developed for finite-state Markov chains.
Extending this framework to continuous-state or non-Markovian dynamics is a natural direction for future work.}

For part (c), the paper quantifies the optimal theoretical accuracy of any Markov chain mixture model in terms of the Kullback-Leibler divergence between mixture components.
The main theoretical contribution is the following theorem, with the proof appearing in Sec.~\ref{sec:analysis}.
\begin{theorem}[Classification error bound]
\label{thm:1}
    Consider any statistical estimator $\hat{Z}$ for the label $Z$ in the Markov chain mixture model 
    \begin{equation*}
    \begin{aligned}
    \mathbb{P}\{Z = i\} &= \mu(i), \\
    \mathbb{P}\{Y_0 = \alpha \,|\, Z = i\} &= \nu_i(\alpha), \\
    \mathbb{P}\{Y_t = \beta \,|\, 
    Z = i, \, Y_{t-1} = \alpha\} &= P_i(\alpha, \beta).
    \end{aligned}
    \end{equation*}
    The estimator can depend on the observed data $\bm{Y}$ and the model parameters but not on the unknown label $Z$.
    Then, the classification error is at least
    \begin{equation*}
        \mathbb{P}\{\hat{Z} \neq Z\}
        \geq \frac{1}{2} \sum_{i=1}^k \max_{j \neq i} \frac{{\rm e}^{-\operatorname{D}_{\rm KL}(\mathbb{P}_i \,\Vert\, \mathbb{P}_j)}}{\mu(i)^{-1} + \mu(j)^{-1}}.
    \end{equation*}
    In this expression, the Kullback-Leibler divergence is defined as
    \begin{equation*}
        \operatorname{D}_{\rm KL}(\mathbb{P}_i \,\Vert\, \mathbb{P}_j) =
        \sum_{\bm{Y}} \mathbb{P}_i(\bm{Y}) \log\biggl(\frac{\mathbb{P}_i(\bm{Y})}{\mathbb{P}_j(\bm{Y})}\biggr).
    \end{equation*}
    The summation runs over all possible trajectories of length $T$, and
    \begin{equation*}
        \mathbb{P}_i(\alpha_0, \ldots, \alpha_T)
        = \nu_i(\alpha_0) \prod_{t=0}^{T-1} P_i(\alpha_t, \alpha_{t+1}).
    \end{equation*}
    is the probability of observing the trajectory $\bm{Y} = (\alpha_0, \ldots, \alpha_T)$ given that $Z = i$.
\end{theorem}
The theorem \edit{
lower bounds the classification error in terms of the Kullback--Leibler (KL) divergence between the trajectory distributions in the mixture model. Conceptually, it is a hypothesis-testing statement: if the mixture includes two similar trajectory distributions, no estimator can reliably assign trajectories to the correct distributions.} 
\edit{The proof lower bounds the classification error by comparing pairs of components and quantifying how well they can be distinguished using $\operatorname{D}_{\rm KL}(\mathbb{P}_i \,\Vert\, \mathbb{P}_j)$. 
The resulting bound exhibits the $e^{-\operatorname{D}_{\rm KL}}$ scaling familiar from the Bretagnolle--Huber inequality~\cite{bretagnolle1979estimation}; see Sec.~\ref{sec:analysis} for further discussion.}

\edit{In the context of Markov chain mixture modeling, the KL divergence in Theorem~\ref{thm:1} is the divergence between \emph{length-$T$ trajectory distributions}. Under mild regularity (e.g., ergodicity), this divergence grows linearly with $T$ at a rate determined by the stationary transition dynamics. Consequently, the lower bound in Theorem~\ref{thm:1} decays exponentially in $T$, highlighting trajectory length as a dominant driver of identifiability.}

\subsection{Plan for the Paper}

The rest of the paper is organized as follows.
Section \ref{sec:msm_extension} introduces Markov state modeling and Markov state mixture modeling,
Section \ref{sec:history} reviews previous related approaches, Section \ref{sec:description} introduces the variational EM algorithm,
Section \ref{sec:analysis} provides theoretical analysis, Section \ref{sec:experiments} presents numerical experiments, and Section \ref{sec:conclude} concludes.

\subsection{Notation}

Scalars are in regular typeface: $k, s, t, T, N, D$.
Lowercase letters $k, s, t$ indicate scalar quantities that can vary from line to line.
Uppercase letters $T, N, D$ indicate scalar quantities that are fixed by the data set.
Finite and infinite sets are in uppercase calligraphic letters, e.g., $\mathcal{X}, \mathcal{Y}$. 
Elements of finite sets are often in lowercase Greek letters, e.g., $\alpha, \beta \in \mathcal{Y}$.
Vectors are typically written in bold lowercase typeface: $\bm{\mu}, \bm{\nu}, \bm{q}$.
Matrices are written in bold uppercase typeface: $\bm{K}, \bm{P}$.
The entries of vectors and matrices are indicated with parentheses and bold typeface: $\mu(i)$, $D(i, i)$.

Trajectories are a special class of vectors indexed by time $t = 0, 1, \ldots, T$ and written in uppercase bold typeface: $\bm{X}, \bm{Y}$.
The states of trajectories are written in regular typeface with subscripts: $X_t, Y_t$.
Superscripts denote trajectories in a collection: $\bm{Y}^1, \bm{Y}^2, \ldots, \bm{Y}^N$.

The probability measure $\mathbb{P}$ represents a ground-truth model,
while the probability measure $\hat{\mathbb{P}}$ is a data-driven estimate.

\section{Modeling: Critiquing the Markov State Model} \label{sec:msm_extension}

A Markov state model is a classic approach for \edit{time-series} analysis \edit{with a long history} \cite{husic2018markov}.
It has recently become popular in chemistry \cite{noe2013projected,roblitz2013fuzzy, kohs2022nonparametric}, biology \cite{chu2017markov,tse2018rare,tan2018exploring,fang2018cell}, and climate science \cite{finkel2023revealing,souza2023transforming,springer2024agnostic}.
A Markov state model approximates \edit{time-series} data as a finite-state Markov chain that jumps between suitably defined ``states'' of the system.

A traditional Markov state model relies on three assumptions that are described in the following sections:
the finite-state assumption (Sec.~\ref{sec:finite_state}), the Markovian assumption (Sec.~\ref{sec:markovian}), and the one-chain assumption (Sec.~\ref{sec:one_chain}).
However, the one-chain assumption may be inappropriate, and Sec.~\ref{sec:extension} will relax the one-chain assumption and generalize the Markov state model to a Markov state mixture model. 

\subsection{The Finite-state Assumption} \label{sec:finite_state}
In a Markov state model, a dynamical system on a general state space $\mathcal{X}$ is reduced to a dynamical system on a finite state space $\mathcal{Y} = \{1, 2, \ldots, s\}$.
The finite states are assumed to be internally homogeneous and meaningfully different from one another.
Clearly delineated states are crucial to the interpretability of the Markov state model.

There are various approaches for defining an appropriate set of finite states.
As one approach, the states can be defined based on \edit{domain expertise}. 
For example, user activity on \texttt{MSNBC.com} (``clickstreams'') can be reduced to a state space $\mathcal{Y}$ consisting of different types of webpages, such as the homepage (state 1), the news pages (state 2), and the sports pages (state 3) \cite{melnykov2016clickclust}.
These state definitions are based on the \texttt{MSNBC.com} site map and the conventional division of articles into ``sports'' versus ``news''. 

As a more systematic alternative, the state space $\mathcal{X}$ can be decomposed using a distance function $d: \mathcal{X} \times \mathcal{X} \rightarrow \mathbb{R}$ and a set of data centers $\{\bm{c}_1, \ldots, \bm{c}_s\}$.
This leads to the division into Voronoi cells
\begin{align*}
    &\mathcal{X} = \mathcal{X}_1 \cup \cdots \cup \mathcal{X}_s, \quad \text{where} \\
    &\mathcal{X}_i = \{\bm{x} \in \mathcal{X} : d(\bm{x}, \bm{c}_i) \leq d(\bm{x}, \bm{c}_j) \text{ for } j \neq i\}.
\end{align*}
Any trajectory $\bm{X}$ on the state space $\mathcal{X}$ can be approximated by a trajectory $\bm{Y}$ on the reduced state space $\mathcal{Y} = \{1, \ldots, s\}$ by setting $Y_t = i$ if $X_t \in \mathcal{X}_i$ with arbitrary tie-breaking at the boundaries.
\edit{
The rest of the paper will work with these discretized trajectories.
}

Identifying a suitable distance function $d$ and a set of data centers $\{\bm{c}_1, \ldots, \bm{c}_s\}$ can be tricky.
When the data lies in $\mathcal{X} = \mathbb{R}^D$,
one popular method combines all the trajectory locations into a data set $\{\bm{x}_i\}_{1 \leq i \leq M}$ and selects the centers $\{\bm{c}_1, \ldots, \bm{c}_s\}$ that minimize the sum of square distances
\begin{equation*}
    \sum_{i=1}^M \min_{1 \leq j \leq s} d(\bm{x}_i, \bm{c}_j)^2.
\end{equation*}
where $d(\bm{x}, \bm{y}) = \lVert \bm{x} - \bm{y} \rVert$ is the Euclidean distance.
This method is called $k$-means clustering \cite{lloyd1982least}.

$k$-means clustering is common but not always effective at identifying finite states.
Spectral clustering \cite{coifman2005geometric} is an alternative method that adapts to the nonlinear manifold structure of the input data.
Spectral clustering can be applied to any state space (not just $\mathcal{X} = \mathbb{R}^D$), as long as there is a kernel function $k:\mathcal{X} \times \mathcal{X} \rightarrow \mathbb{R}$ which quantifies the similarity between data points.
\fig~\ref{fig:clustering} illustrates spectral clustering on a spiral data set using the Gaussian kernel $k(\bm{x}, \bm{y}) = \exp(- \tfrac{1}{2 \sigma^2} \lVert \bm{x} - \bm{y} \rVert^2)$ with bandwidth $\sigma = 1$. \edit{Spectral clustering is now standard practice in Markov state modeling, so the algorithmic details are omitted; see \cite{coifman2005geometric} for the classical formulation and \cite{chen2023randomly} for a scalable implementation based on randomized linear algebra.
Later, Sec.~\ref{sec:gene} will apply spectral clustering to construct Markov states for gene expression data.}

\begin{figure}[t]
\centering
\includegraphics[width=\linewidth]{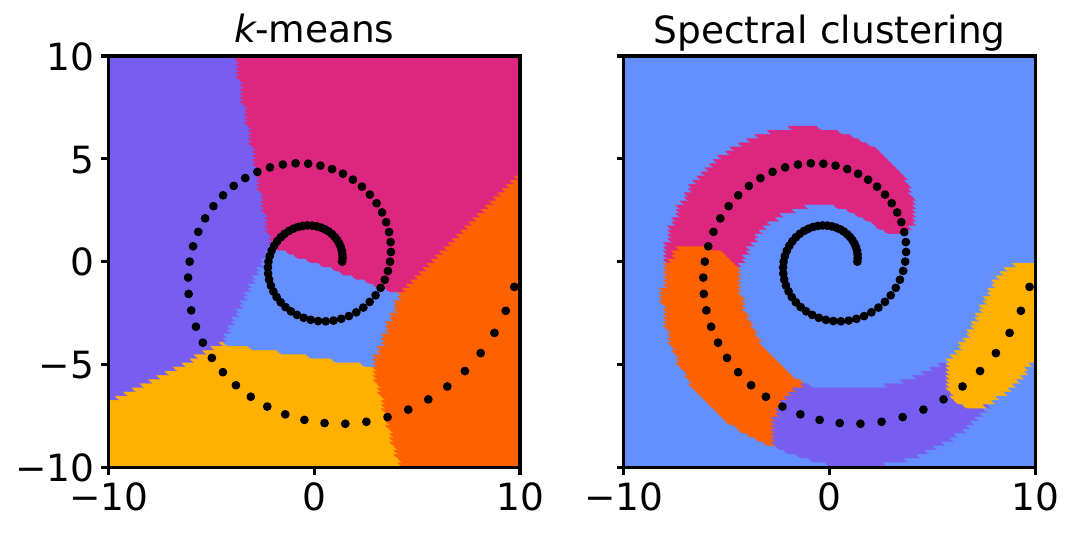}
\caption{Comparison of $k$-means (left) versus spectral clustering (right) applied to $M = 100$ data points in a spiral.
Spectral clustering adapts to the nonlinear manifold structure of the data, $k$-means does not. \label{fig:clustering}}
\end{figure}

\subsection{The Markovian Assumption} \label{sec:markovian}

In a Markov state model, each trajectory $\bm{Y}$ is assumed to be \emph{Markovian}.
The Markovian assumption means that the position $Y_t$ depends only on the immediate past position $Y_{t-1}$, regardless of all the earlier positions $Y_{t-2}, Y_{t-3}, \ldots$.
Thus, all the transition probabilities can be written as
\begin{align*}
    &\mathbb{P}\{Y_t = \beta \,|\, Y_0 = \alpha_0, \ldots, Y_{t-1} = \alpha_{t-1} \} \\
    &= P(\alpha_{t-1}, \beta),
\end{align*}
for a suitable transition matrix $\bm{P} \in \mathbb{R}^{s \times s}$.

Like all model assumptions, the Markovian assumption is an idealization of the truth.
Many dynamical systems are not precisely Markovian because they exhibit prolonged history effects \cite{chierichetti2012are}.
As an example of history dependence, consider the travel patterns of adults in Nanjing, China \cite{zhou2021you}. 
On a given day, a parent is unlikely to commute to their child's school more than twice (first taking the child to school, second picking the child up from school).
Because of this history dependence, the location of the adult during the day does not precisely form a Markov chain.
Nonetheless, a rich and detailed description of the daily lives of residents is possible using the Markov chain assumption \cite{zhou2021you}.
In general, a Markov state model assumes that history effects can be disregarded while still providing a meaningful analysis.

\edit{When strict Markovianity fails due to finite memory, a standard remedy is \emph{delay embedding}: redefine the state at time $t$ as a tuple of lagged observations, e.g., $(Y_t,Y_{t-1},\ldots,Y_{t-m+1})$. Delay embedding converts an $m$th-order Markov chain into a first-order Markov chain on an enlarged state space.
Markov state modeling can be applied directly to the delay-embedded process, but the trade-off is an increased state-space size and hence more parameters to estimate.}

\subsection{The One-Chain Assumption} \label{sec:one_chain}

The final assumption of a Markov state model is rarely stated yet fundamental.
The standard presentations of Markov state modeling \cite{pande2010everything,chodera2014markov,husic2018markov} all assume that a single Markov chain generates the trajectories.
If the data set consists of multiple trajectories $\bm{Y}^1, \bm{Y}^2, \ldots, \edit{\bm{Y}^N}$,
the trajectories are all assumed to be generated from the same transition matrix $\bm{P}$.
The one-chain assumption is well-justified if the trajectories are computer simulations from a fixed model.
Yet, the one-chain assumption 
is challenged by the heterogeneities of real-world observations and experiments. 
For example, is there really one type of web user on \texttt{MSNBC.com}?
Is there really one type of commuter in Nanjing, China?
Several papers in the statistics literature go beyond the one-chain assumption and instead they analyze \texttt{MSNBC.com} \cite{melnykov2016clickclust} users and Nanjing commuters \cite{zhou2021you} using a \emph{mixture} of Markov chains with distinct transition matrices $\bm{P}_1, \ldots, \bm{P}_k$.

The resulting Markov chain mixture model will be described in the next section, and it will be the focus of the remainder of the paper.
Although joint inference of states and Markov chain mixture modeling parameters is possible \cite{safinianaini2024copymix},
the paper assumes the state sequences have already been defined using domain knowledge or data-driven clustering. 
This design choice decouples state assignment from mixture inference
to support scalable, domain-adaptable modeling.

\subsection{The Markov chain mixture model} \label{sec:extension}

In a Markov chain mixture model, each trajectory $\bm{Y}$ is generated as follows:
\begin{equation}
\label{eq:the_model}
\begin{aligned}
    \mathbb{P}\{Z = i\} &= \mu(i), \\
    \mathbb{P}\{Y_0 = \alpha \,|\, Z = i\} &= \nu_i(\alpha), \\
    \mathbb{P}\{Y_t = \beta \,|\, 
    Z = i, \, Y_{t-1} = \alpha\} &= P_i(\alpha, \beta).
\end{aligned}
\end{equation}
The variable $Z \in \{1, \ldots, k\}$ is a latent variable indicating the Markov chain label.
$\mu(i)$ describes the probability of observing each label $Z = i$.
The vectors $\bm{\nu}_i \in [0, 1]^s$ and matrices $\bm{P}_i \in [0, 1]^{s \times s}$ give the initialization probabilities and transition probabilities for the Markov chain with label $Z = i$.
The total number of free parameters in the model is $k s^2 - 1$, and typical parameter values are $k = 10^0$--$10^2$ and $s = 10^0$--$10^2$.
The Markov chain mixture model was presented in the papers \cite{smyth1997clustering,ramoni2002bayesian,batu2004inferring}, and it has been applied in areas including consumer behavior \cite{melnykov2016clickclust,zhou2021you}, biology \cite{das2023clustering, das2023utilizing}, psychology \cite{de2017use}, and economics \cite{frydman2005estimation}.

\section{Algorithms: Comparing Approaches for Markov Chain Mixtures} \label{sec:history}

This section discusses past algorithms for 
fitting the parameters in a Markov chain mixture model, and it emphasizes three main themes.
First, many researchers rely on the expectation-maximization (EM) algorithm for parameter inference (Sec.~\ref{sec:em}).
Second, researchers have identified the key importance of the trajectory length $T$ for ensuring classification accuracy (Sec.~\ref{sec:length}).
Third, this paper adapts the variational EM algorithm \cite[Sec.~10]{bishop2006pattern} instead of traditional EM since it learns the number of mixture components automatically (Sec.~\ref{sec:description}).

\subsection{The Expectation-Maximization Algorithm} \label{sec:em}

In several applications of Markov chain mixture modeling, the expectation-maximization (EM) algorithm gives fast and reliable parameter estimates \cite{smyth1997clustering,ramoni2002bayesian,melnykov2016clickclust,zhou2021you}, especially when it is initialized carefully \cite{kausik2023learning,GKV16,spaeh2023learning}.
\edit{In EM, each iteration alternates between (i) estimating conditional probabilities for the latent variables $\mathbb{P}\{Z^n=i\mid \bm{Y}^n\}$ and (ii) updating parameters $(\bm{\mu},\bm{\nu}_i,\bm{P}_i)$ using the expected counts of initial states and transitions under the latent variable distribution.
Because EM for finite-state Markov-chain mixtures is standard (see, e.g., \cite{dempster1977maximum,kausik2023learning}), this paper provides only a brief summary and focuses instead on the detailed exposition of variational EM in Sec.~\ref{sec:description}.}

As a limitation, EM is a nonconvex optimization algorithm that converges to local maxima in the likelihood landscape.
Therefore, the accuracy of the converged parameters depends on the initialization, which can be chosen either randomly \cite{melnykov2016clickclust,zhou2021you} or by a more structured approach \cite{smyth1997clustering,GKV16,spaeh2023learning,kausik2023learning}.
In a random initialization, the user can generate probability estimates $\hat{\PP}\{Z^n = i\}$ for $i = 1, \ldots, k$ that are uniformly distributed on the probability simplex.
A convenient procedure for generating uniformly distributed vectors on the simplex is to first produce $k$ independent exponential random variables and then normalize them to sum to one \cite{onn2011generating}.

As another limitation,
EM does not automatically determine the number of mixture components.
Rather, the number of components must be specified as input.
Many users fit several models with different numbers of components and then perform a comparison using AIC, BIC, or cross-validation \cite{keribin2000consistent,smyth1997clustering,melnykov2016clickclust}.
When the user searches over a space of $k$ different models,
this procedure can increase the runtime by a factor of $k$ or more.
Sec.~\ref{sec:description} will present an improvement of the EM algorithm called variational EM (VEM) that \emph{does} automatically determine the number of components.

\subsection{\texorpdfstring{The Trajectory Length $T$}{The Trajectory Length T}} \label{sec:length}

The second theme in the Markov chain mixture modeling literature is the close link between the classification accuracy and the trajectory length $T$.
Put simply, classifying short trajectories with the correct labels is \emph{hard}.
For example, Ramoni and coauthors \cite[Sec.~5.3.1]{ramoni2002bayesian} describe numerical experiments with increasing trajectory lengths that lead to steadily increasing accuracy.

Researchers have tackled the issue of trajectory length by focusing on the extremes of short trajectories ($T = 3$, \cite{GKV16,spaeh2023learning}) or long trajectories ($T \rightarrow \infty$, \cite{khaleghi2016consistent,fitzpatrick2022asymptotics,kausik2023learning}).
A couple papers propose structured initialization methods for EM based on short trajectories of length $T = 3$ \cite{GKV16,spaeh2023learning}, but these approaches require massive amounts of data, for example, $N = 10^8$ trajectories to achieve classification error $< 10\%$ \cite{GKV16}.
At the other extreme, algorithms can ensure perfect parameter identification in the limit $T \rightarrow \infty$ \cite{khaleghi2016consistent,kausik2023learning}.

Sec.~\ref{sec:analysis} provides further insight into the intermediate regimes where $T$ is neither short nor long, by bounding the classification error for any Markov chain mixture model with any trajectory length $T$.
\edit{Since the KL divergence grows roughly linearly with $T$, the lower bound on the classification error decays roughly exponentially in $T$, indicating the importance of long trajectories.}

\subsection{Variational EM Algorithm for Markov Chain Mixtures} \label{sec:description}

The current section describes the variational EM method for Markov chain mixtures in detail and outlines the steps for a user to apply the method to data.

As the first conceptual step, the user must reduce the data to a finite state space $\mathcal{Y} = \{1, \ldots, s\}$.
The states can be chosen based on \edit{domain expertise} or based on spectral clustering \cite{coifman2005geometric}, as is common in the Markov state modeling literature.
The typical number of states is $s = 10^0$--$10^2$.

Next, the user must choose the maximum number of mixture components, typically $k = 10^1$--$10^2$.
As long as $k \geq k_\true$, the approach is capable of identifying the true number of components.
However, there is motivation to not set $k$ too large, as the computational cost is proportional to $k$.

The data set is then described using the Bayesian mixture model:
\begin{equation}
\label{eq:new_model}
\begin{aligned}
    \bm{\mu} &\sim \operatorname{Dir}(\bm{1}_k / k), \\
    \bm{\nu}_i &\sim \operatorname{Dir}(\bm{1}_s), \\
    \bm{P}_i(\alpha\, \cdot) &\sim \operatorname{Dir}(\bm{1}_s ), \\    
    \mathbb{P}\{Z = i\} &= \mu(i), \\
    \mathbb{P}\{Y_0 = \alpha \,|\, Z = i\} &= \nu_i(\alpha), \\
    \mathbb{P}\{Y_t = \beta \,|\, 
    Z = i, \, Y_{t-1} = \alpha\} &= P_i(\alpha, \beta).
\end{aligned}
\end{equation}
\edit{This Bayesian formulation combines} standard Bayesian priors for Markov state modeling \cite{bacallado2009bayesian} and mixture modeling \cite{gormley2023model}.
Specifically, the model places uninformative $\operatorname{Dir}(\bm{1}_s)$ prior distributions on the parameters $\bm{\nu}_i$ and the rows of the matrix $\bm{P}_i(\alpha, \cdot)$, and it places a sparsity-promoting $\operatorname{Dir}(\bm{1}_k / k)$ prior distribution on the mixture parameter $\bm{\mu}$.
The $\operatorname{Dir}(\bm{1}_k / k)$  prior leads to a posterior distribution that automatically prunes extraneous mixture components by setting the corresponding mixture probabilities to $\mu(i) \approx 0$ \cite{ishwaran2001bayesian,rousseau2011asymptotic,malsiner2014model}.

Given the model \eqref{eq:new_model} and the observed trajectories, the Bayesian posterior distribution is hard to compute exactly.
However, the posterior can be efficiently approximated
by a factored distribution (see \cite[Sec.~10.2]{bishop2006pattern}):
\begin{equation*}
    \mathbb{P}(\{Z^n, \bm{\mu}, \bm{\nu}_i, \bm{P}_i\}) \approx Q_1(\{Z^n\}) \, Q_2(\{\bm{\mu}, \bm{\nu}_i, \bm{P}_i\}).
\end{equation*}
This factorization assumes the latent variables $\{Z^n\}$ are independent from the mixture parameters $\{\bm{\mu}, \bm{\nu}_i, \bm{P}_i \}$.
Given this factorization, the approximation minimizing the Kullback-Leibler divergence with respect to the posterior is a sequence of mutually independent random variables
\begin{equation}
\label{eq:the_posteriors}
\begin{aligned}
    \bm{\mu} &\sim \operatorname{Dir}(\bm{N}), \\
    \bm{\nu}_i &\sim \operatorname{Dir}(\bm{N}_i), \\
    \bm{P}_i(\alpha, \cdot) &\sim \operatorname{Dir}(\bm{N}_{i,\alpha}).
\end{aligned}
\end{equation}
with parameters that can be optimized using the variational EM algorithm in Alg.~\ref{alg:variational}.

\begin{algorithm}[p]
\caption{Variational EM Algorithm \label{alg:variational}}
\begin{algorithmic}[1]
\Require Trajectories $\bm{Y}^n$, estimates $\hat{\PP}\{Z^n = i\}$.
\Ensure Refined estimates $\hat{\bm{N}}, \hat{\bm{N}}_i, \hat{\bm{N}}_{i,\alpha}, \hat{Z}^n$.
\State
For each $n = 1, \ldots, N$, form
\begin{align*}
    U^n(\alpha) &= \mathds{1}\{Y_0^n = \alpha\}, \\
    V^n(\alpha, \beta) &= \sum_{t = 0}^{T-1} \mathds{1}\{Y_t^n = \alpha, \, Y_{t+1}^n = \beta\}.
\end{align*}
\For{$t = 1, 2, \ldots$}
\State Compute Dirichlet parameters
\begin{align*}
    \hat{N}(i) 
    &= \frac{1}{k} + \sum_{n=1}^N \hat{\PP}\{Z^n = i\}, \\
    \hat{N}_i (\alpha) 
    &= 1 + \sum_{n=1}^N \hat{\PP}\{Z^n = i\} \,U^n(\alpha), \\
    \hat{N}_{i, \alpha} (\beta) 
    &= 1 + \sum_{n=1}^N\, \hat{\PP}\{Z^n = i\} V^n(\alpha, \beta).
\end{align*}
\State Compute vectors $\hat{\bm{\mu}}$, $\hat{\bm{\nu}}_i$, $\hat{\bm{P}}_i(\alpha, \cdot)$ via
\begin{align*}
    \log \hat{\mu}(i)
    &= \psi(\hat{N}(i)) - \psi\Bigl(\sum_{j=1}^k \hat{N}(j) \Bigr), \\
    \log \hat{\nu}_i(\alpha) 
    &= \psi(\hat{N}_i(\alpha)) - \psi \Bigl(\sum_{\beta=1}^s \hat{N}_i(\beta)\Bigr), \\
    \log \hat{P}_i(\alpha, \beta) 
    &= \psi(\hat{N}_{i, \alpha}(\beta)) - \psi \Bigl(\sum_{\gamma=1}^s \hat{N}_{i, \alpha}(\gamma) \Bigr).
\end{align*}
$\psi(x) = \frac{d}{dx} \log  \Gamma(x)$ is the digamma function.
\State Update probabilities $\hat{\PP}\{Z^n = i\}$ via
\begin{equation*}
    \frac{\hat{\mu}(i)}{C_n}
    \prod_{\alpha = 1}^s \hat{\nu}_i(\alpha)^{U^n(\alpha)} 
    \prod_{\alpha, \beta = 1}^s \hat{P}_i(\alpha, \beta)^{V^n(\alpha, \beta)},
\end{equation*}
and choose $C_n$ so that $\sum_{1 \leq i \leq k} \hat{\mathbb{P}}\{Z^n = i\} = 1$.
\State Calculate variational lower bound
\begin{align*}
    \mathcal{L} 
    &= \sum_{n=1}^N \log C_n + \log \biggl[\frac{B(\bm{1}_k/k)}{B(\hat{\bm{N}})} \prod_{i=1}^k \hat{\mu}_i^{\hat{N}(i) - 1/k} \biggr] \\
    &+ \sum_{i=1}^k \log \biggl[\frac{B(\bm{1}_s)}{B(\hat{\bm{N}}_i)} \prod_{\alpha=1}^s \edit{\hat{\nu}_i(\alpha)}^{\hat{N}_i(\alpha) - 1} \biggr] \\
    &+ \sum_{i=1}^k \sum_{\alpha=1}^s
    \log \biggl[\frac{B(\bm{1}_s)}{B(\hat{\bm{N}}_{i,\alpha})}
    \prod_{\beta=1}^s \edit{\hat{P}_i(\alpha, \beta)}^{\hat{N}_{i, \alpha}(\beta) - 1}\Biggr],
\end{align*}
where $B(\bm{N}) = (\prod_i \Gamma(N_i)) \,\slash\, \Gamma (\sum_i N_i)$.
\State \textbf{Break} if $|\Delta \mathcal{L}|< 10^{-12} \cdot N T$.
\EndFor
\State Return $\hat{\bm{N}}, \hat{\bm{N}}_i, \hat{\bm{N}}_{i,\alpha}, \hat{Z}^n = \operatorname*{argmax}\limits_{1 \leq i \leq k} \hat{\mathbb{P}}\{Z^n = i\}$.
\end{algorithmic}
\end{algorithm}

Examining Alg.~\ref{alg:variational}, there are a couple of key changes from the standard EM algorithm (Sec.~\ref{sec:em}).
First, the parameters $\{\bm{\mu}, \bm{\nu}_i, \bm{P}_i\}$ are described by Dirichlet distributions, which encode point estimates and uncertainty intervals.
The user can extract a point estimate for any parameter by using the expected value formula for the Dirichlet distribution, e.g.,
\begin{equation}
\label{eq:point}
    \hat{\mu}(i) = \mathbb{E}[\mu(i)] = \frac{N(i)}{\sum_{j=1}^k N(j)}.
\end{equation}
The user can extract uncertainty intervals by using the variance formula for the Dirichlet distribution, e.g.,
\begin{equation*}
    \operatorname{Var}[\mu(i)] = \frac{N(i) \sum_{j \neq i} N(j)}{\bigl(\sum_{j=1}^k N(j)\bigr)^2 \bigl(\sum_{j=1}^k N(j) + 1\bigr)}.
\end{equation*}
The automatic availability of uncertainty information is a helpful feature of the Bayesian approach.

As the second change, variational EM automatically determines the number of mixture components.
Upon convergence, only a small number of mixture parameters satisfy $\hat{N}(i) \geq 1$.
The optimization drives all the remaining mixture parameters to the lowest possible value $\hat{N}(i) = 1/k$.
These components are not represented by any labels $\hat{Z}^n = i$, so they are effectively removed from the model.
The tendency of variational EM to identify a parsimonious mixture model is called \emph{automatic
relevance determination}.
It was observed in the papers \cite{corduneanu2001hyperparameters,mackay2001local} and later justified based on asymptotic arguments in \cite{rousseau2011asymptotic}.

To apply Alg.~\ref{alg:variational} in practice, the user needs some way to initialize the parameter estimates.
Since there is limited theoretical understanding of the best way to initialize variational EM,
this paper randomly initializes parameters $\hat{\PP}\{Z^n = i\}$ for $i = 1, \ldots, k$ from the probability simplex.
After running Alg.~\ref{alg:variational} with $10^2$ independent initializations, the paper selects the optimal parameters $\{\hat{\bm{N}}, \hat{\bm{N}}_i, \hat{\bm{N}}_{i,\alpha}, \hat{Z}^n\}$ that maximize the likelihood lower bound $\mathcal{L}$.

\section{Theory: Limits on Markov Chain Mixture Modeling} \label{sec:analysis}

Markov chain mixture modeling cannot succeed when the trajectories are too short.
This section provides a theoretical explanation by proving and interpreting Theorem.~\ref{thm:1}.

\begin{proof}[Proof of Thm.~\ref{thm:1}]
    Given the observed data $\bm{Y}$, the conditional distribution of the label $Z$ is
    \begin{equation}
    \label{eq:use_me}
        \mathbb{P}\{Z = i \,|\, \bm{Y}\} = \frac{\mu(i) \mathbb{P}_i(\bm{Y})}{\sum_{j=1}^k \mu(j) \mathbb{P}_j(\bm{Y})}
    \end{equation}
    for each $i = 1, \ldots, k$.
    The optimal classification algorithm selects the labels according to
    \begin{equation*}
        \hat{Z} \in \operatorname*{argmax}_{1 \leq i \leq k} \mathbb{P}\{Z = i \,|\, \bm{Y}\},
    \end{equation*}
    and the resulting statistical error is quantified by
    \begin{equation*}
        \mathbb{P}\{Z \neq \hat{Z} \,|\, \bm{Y}\} 
        = 1 - \max_{1 \leq i \leq k} \mathbb{P}\{Z = i \,|\, \bm{Y}\}.
    \end{equation*}
    Next, use the inequality $1 - \alpha \geq (1 - \alpha^2)/2$ to obtain the lower bound
    \begin{align*}
        & \mathbb{P}\{Z \neq \hat{Z} \,|\, \bm{Y}\} \\
        &= 1 - \max_{1 \leq i \leq k} \mathbb{P}\{Z = i \,|\, \bm{Y}\} \\
        &\geq \frac{1}{2} - \frac{1}{2} \max_{1 \leq i \leq k} \mathbb{P}\{Z = i \,|\, \bm{Y}\}^2 \\
        &\geq \frac{1}{2} - \frac{1}{2} \sum_{i=1}^k \mathbb{P}\{Z = i \,|\, \bm{Y}\}^2 \\
        &= \frac{1}{2} \sum_{i=1}^k 
        \mathbb{P}\{Z = i \,|\, \bm{Y}\}
        (1 - \mathbb{P}\{Z = i \,|\, \bm{Y}\} ).
    \end{align*}
    Multiply both sides by $\mathbb{P}(\bm{Y})$ and sum over all possible trajectories $\bm{Y}$ to obtain the formula
    \begin{align*}
        & \mathbb{P}\{Z \neq \hat{Z}\} 
        = \sum_{\bm{Y}} \mathbb{P}\{Z \neq \hat{Z} \,|\, \bm{Y}\}\,\mathbb{P}(\bm{Y}) \\
        &\geq \frac{1}{2} \sum_{i, \bm{Y}} \mathbb{P}\{Z = i, \bm{Y}\} (1 - \mathbb{P}\{Z = i \,|\, \bm{Y}\} ) \\
        &= \frac{1}{2} \sum_{i=1}^k \mu(i) \frac{\sum_{\bm{Y}} \mathbb{P}\{Z = i, \bm{Y}\} (1 - \mathbb{P}\{Z = i \,|\, \bm{Y}\})}{\mu(i)}.
    \end{align*}
    The latter quantity is a conditional expectation
    \begin{equation}
        \label{eq:formula}
        \begin{aligned}
        & \mathbb{P}\{Z \neq \hat{Z}\} \\
        &\geq \frac{1}{2} \sum_{i=1}^k \mu(i) \,\mathbb{E}\bigl[1 - \mathbb{P}\{Z = i \,|\, \bm{Y}\} \,|\, Z = i\bigr]. 
    \end{aligned}
    \end{equation}
    The conditional expectation turns out to be analytically tractable and allows us to derive the bound directly.

    The rest of the proof is based on bounding $\mathbb{E}\bigl[1 - \mathbb{P}\{Z = i \,|\, \bm{Y}\} \,|\, Z = i\bigr]$ from below using the Kullback-Leibler divergence.
    To obtain a convenient lower bound, use eq.~\eqref{eq:use_me} to derive
    \begin{align*}
        & \mathbb{E}\bigl[1 - \mathbb{P}\{Z = i \,|\, \bm{Y}\} \,|\, Z = i\bigr] \\
        &= \mathbb{E}\Biggl[ \frac{\sum_{j \neq i} \mu(j) \mathbb{P}_j(\bm{Y})}{\sum_{j=1}^k \mu(j) \mathbb{P}_j(\bm{Y})} \,\bigg|\, Z = i\Biggr] \\
        &\geq \mathbb{E}\biggl[ \frac{\mu(j) \mathbb{P}_j(\bm{Y})}{\mu(i) \mathbb{P}_i(\bm{Y}) + \mu(j) \mathbb{P}_j(\bm{Y})} \,\bigg|\, Z = i\biggr],
    \end{align*}
    for each $j \neq i$.
    The inequality is due to the fact that $f(x) = x / (c + x)$ is an increasing function of $x > 0$, for any $c > 0$.
    Next, by an application of Jensen's inequality,
    \begin{align*}
        & \log\biggl(\mathbb{E} \biggl[ \frac{\mu(j) \mathbb{P}_j(\bm{Y})}{\mu(i) \mathbb{P}_i(\bm{Y}) + \mu(j) \mathbb{P}_j(\bm{Y})} \,\bigg|\, Z = i\biggr]\biggr) \\
        &\geq \mathbb{E} \biggl[ \log\biggl(\frac{\mu(j) \mathbb{P}_j(\bm{Y})}{\mu(i) \mathbb{P}_i(\bm{Y}) + \mu(j) \mathbb{P}_j(\bm{Y})}\biggr) \,\bigg|\, Z = i \biggr] \\
        &\begin{aligned}
        &= \mathbb{E}\biggl[\log\biggl( \frac{\mathbb{P}_j(\bm{Y})}{\mathbb{P}_i(\bm{Y})}\biggr) \,\bigg|\, Z = i\biggr] \\ 
        &\quad - \mathbb{E}\biggl[\log\biggl(\frac{\mu(i)}{\mu(j)} + \frac{\mathbb{P}_j(\bm{Y})}{\mathbb{P}_i(\bm{Y})} \biggr) \,\bigg|\, Z = i\biggr].
        \end{aligned}
    \end{align*}
    The first term in the final expression is the negative Kullback-Leibler divergence
    \begin{equation*}
        \mathbb{E}\biggl[\log\biggl( \frac{\mathbb{P}_j(\bm{Y})}{\mathbb{P}_i(\bm{Y})}\biggr) \,\bigg|\, Z = i\biggr] = -\operatorname{D}_{\rm KL}(\mathbb{P}_i \,\Vert\, \mathbb{P}_j).
    \end{equation*}
    The second term can be bounded by another application of Jensen's inequality:
    \begin{align*}
        & -\mathbb{E}\biggl[\log\biggl(\frac{\mu(i)}{\mu(j)} + \frac{\mathbb{P}_j(\bm{Y})}{\mathbb{P}_i(\bm{Y})} \biggr) \,\bigg|\, Z = i\biggr] \\
        &\geq -\log\biggl(\frac{\mu(i)}{\mu(j)} + 
        \mathbb{E}\biggl[ \frac{\mathbb{P}_j(\bm{Y})}{\mathbb{P}_i(\bm{Y})} \,\bigg|\, Z = i\biggr] \biggr).
    \end{align*}
    Moreover, observe that
    \begin{align*}
        \mathbb{E}\biggl[ \frac{\mathbb{P}_j(\bm{Y})}{\mathbb{P}_i(\bm{Y})} \,\bigg|\, Z = i\biggr]
        &= \sum_{\bm{Y}} \mathbb{P}_i(\bm{Y}) \frac{\mathbb{P}_j(\bm{Y})}{\mathbb{P}_i(\bm{Y})} \\
        &= \sum_{\bm{Y}} \mathbb{P}_j(\bm{Y}) = 1.
    \end{align*}
    Tying together the above expressions, it follows
    \begin{equation*}
    \begin{aligned}
        & \log\bigl(\mathbb{E}\bigl[1 - \mathbb{P}\{Z = i \,|\, \bm{Y}\} \,|\, Z = i\bigr] \bigr) \\
        &\geq -\operatorname{D}_{\rm KL}(\mathbb{P}_i \,\Vert\, \mathbb{P}_j) - \log\biggl(\frac{\mu(i)}{\mu(j)} + 
        1 \biggr)
    \end{aligned}        
    \end{equation*}
    and therefore
    \begin{equation*}
        \mathbb{E}\bigl[1 - \mathbb{P}\{Z = i \,|\, \bm{Y}\} \,|\, Z = i\bigr]
        \geq \max_{j \neq i} \frac{{\rm e}^{-\operatorname{D}_{\rm KL}(\mathbb{P}_i \,\Vert\, \mathbb{P}_j)}}{\frac{\mu(i)}{\mu(j)} + 1}.
    \end{equation*}
    Combine with eq.~\eqref{eq:formula} to conclude that the classification error is at least
    \begin{equation*}
        \mathbb{P}\{\hat{Z} \neq Z\}
        \geq \frac{1}{2} \sum_{i=1}^k \max_{j \neq i} \frac{{\rm e}^{-\operatorname{D}_{\rm KL}(\mathbb{P}_i \,\Vert\, \mathbb{P}_j)}}{\mu(i)^{-1} + \mu(j)^{-1}}.
    \end{equation*}
    This completes the proof.
\end{proof}

\edit{Two points are worth emphasizing.
First, Theorem~\ref{thm:1} is non-asymptotic: it applies to any finite trajectory length $T$ without requiring mixing or stationarity.
It thus provides a baseline limit in the ``many short trajectories'' regime, isolating an intrinsic identifiability barrier. 

Second, the $e^{-\operatorname{D}_{\rm KL}}$ scaling in the error bound reflects a classical connection between classification error and information divergences.
The optimal error for distinguishing $\mathbb{P}_i$ from $\mathbb{P}_j$ is
\begin{equation*}
    \sum_{\bm{Y}} \min\{\bm{\mathbb{P}}_i(\bm{Y}), \bm{\mathbb{P}}_j(\bm{Y})\} = 1 - d_{\rm TV}(\mathbb{P}_i, \mathbb{P}_j),
\end{equation*}
where the total variation distance is
\begin{equation*}
    d_{\rm TV}(\mathbb{P}_i, \mathbb{P}_j) = \tfrac{1}{2} \sum_{\bm{Y}} |\mathbb{P}_i(\bm{Y}) - \mathbb{P}_j(\bm{Y})|.
\end{equation*}
The Bretagnolle--Huber inequality~\cite{bretagnolle1979estimation} then bounds 
\begin{equation*}
1 - d_{\rm TV}(\mathbb{P}_i, \mathbb{P}_j) \geq \tfrac{1}{2} e^{-\operatorname{D}_{\rm KL}(\mathbb{P}_i \,\Vert\, \mathbb{P}_j)}
\end{equation*}
Thm.~\ref{thm:1} derives a comparable bound for the Markov mixture, making explicit how the KL divergence controls the overall classification error.
This result complements recent work by Lee et al.~\cite{lee2025near}, who prove high-probability error bounds and also find that KL divergence governs the clustering error.}

To interpret the theorem, it is instructive to consider what happens to the KL divergence in the specific case of a Markov chain mixture as the trajectory length $T$ approaches infinity.
If the Markov chain $i$ has a unique stationary measure $\bm{\pi}_i \in [0, 1]^s$, then in the limit $T \rightarrow \infty$ the ergodic theorem for Markov chains \cite[Sec.~1.10]{norris1997markov} guarantees
\begin{align*}
    &\tfrac{1}{T} \operatorname{D}_{\rm KL}(\mathbb{P}_i \,\Vert\, \mathbb{P}_j) \\
    &= \frac{1}{T} \sum_{\alpha_0, \ldots, \alpha_T} \mathbb{P}_i(\alpha_0, \ldots, \alpha_T) \log\biggl(\frac{\mathbb{P}_i(\alpha_0, \ldots, \alpha_T)}{\mathbb{P}_j(\alpha_0, \ldots, \alpha_T)}\biggr) \\
    &\rightarrow \sum_{\alpha, \beta = 1}^s \edit{\pi_i(\alpha)}
    \edit{P_i(\alpha, \beta)} \log\biggl(\frac{\edit{P_i(\alpha, \beta)}}{\edit{P_j(\alpha, \beta)}}\biggr) \\
    &= \sum_{\alpha=1}^s \edit{\pi_i(\alpha)} \operatorname{D}_{\rm KL}(\edit{P_i(\alpha, \cdot)} \,\Vert\, \edit{P_j(\alpha, \cdot)}).
\end{align*}
By this derivation, the KL divergence $\operatorname{D}_{\rm KL}(\mathbb{P}_i \,\Vert\, \mathbb{P}_j)$ is approximately $T$ times the averaged Kullback-Leibler divergence $\operatorname{D}_{\rm KL}(\edit{P_i(\alpha, \cdot)} \,\Vert\, \edit{P_j(\alpha, \cdot)})$ over all the states.
\edit{The linear dependence of the Kullback-Leibler divergence on $T$ has a practical implication: longer trajectories help substantially. Doubling $T$ can turn a hard problem into an easy one, for example reducing error from $10\%$ to $1\%$.}

\begin{figure*}[t]
    \centering
    \includegraphics[width=0.7\linewidth]{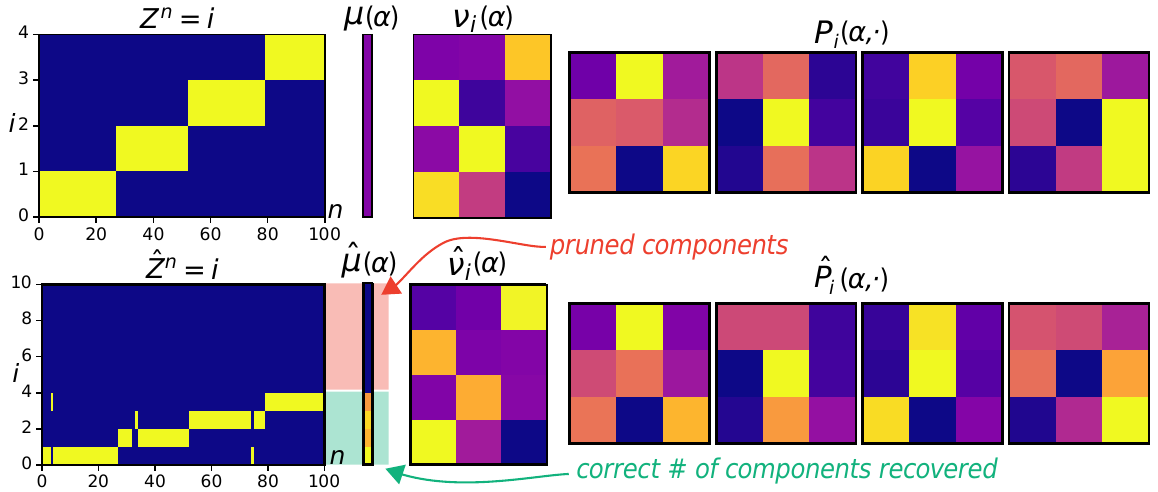}
    \caption{True parameters and parameter estimates for synthetic trajectories with $k_\true=4$, $s=3$, $N=100$, and $T=30$. 
    Top row shows true parameters.
    Bottom row shows parameter estimates from variational EM with a maximum of $k=10$ components. 
    Even with limited data, VEM correctly identifies the number of components.}
    \label{fig:fig_ex}
\end{figure*}

\section{Numerical Experiments} \label{sec:experiments}

This section describes four numerical experiments with the variational EM algorithm.
The first experiment (Sec.~\ref{sec:synthetic}) uses synthetic Markov chain data.
The middle two experiments (Secs.~\ref{sec:lastfm} and ~\ref{sec:ultrarunners}) use real-world observations.
The last experiment (Sec.~\ref{sec:gene}) uses synthetic genetics data.
Code to run the experiments is available at \url{https://github.com/chris-miles/MarkovChain-VEM}.

\subsection{Synthetic Markov Chain Data} \label{sec:synthetic}

The first experiment is based on simulations of the Markov chain mixture model \eqref{eq:the_model} with uniformly random parameters $\bm{\mu}$, $\bm{\nu}_i$, and $\bm{P}_i$.
The goal of the experiment is to test whether variational EM correctly identifies the number of components and correctly classifies trajectories into components.

As the first result, variational EM accurately identifies the number of components even with a relatively small data set.
\fig\ref{fig:fig_ex} shows the true parameters (top) and parameter estimates (bottom) when $k_\true=4$, $s=3$, $N=100$, and $T=30$.
Although variational EM is run with a maximum number of $k = 10$ components, it leads to a fitted model with just 4 components.
The algorithm assigns no trajectories to 6 of the 10 components, thereby pruning them from the model.
The automatic identification of the number of mixture components is a major advantage compared to previous Markov chain mixture modeling approaches \cite{keribin2000consistent,smyth1997clustering,melnykov2016clickclust}.

The reader may wonder, is there a cost to setting the value of $k$ excessively high?
While there is an obvious computational burden to setting $k>k_{\true}$, 
the results in \fig~\ref{fig:accuracy} show that the classification accuracy, defined by
\begin{equation}
\label{eq:acc}
    \operatorname{acc} = \frac{1}{N} \sum_{n=1}^N \mathds{1}\{Z^n = \hat{Z}^n\},
\end{equation}
is not hindered. 
Indeed, the results with $k = 4 = k_\true$ (top panel) are nearly identical to results with $k = 10 > k_\true$ (bottom panel).

\begin{figure}[t]
    \centering
    \includegraphics[width=0.8\linewidth]{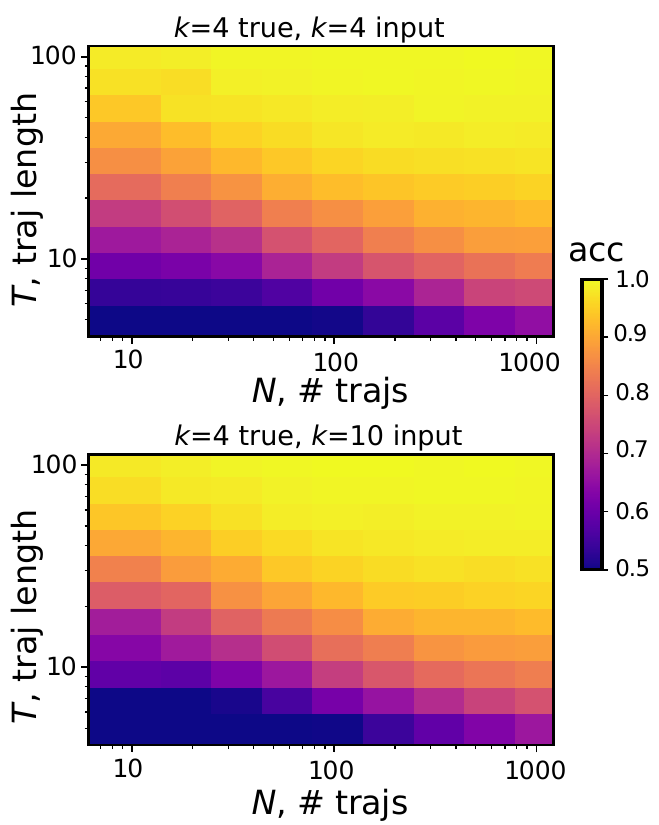}
    \caption{Classification accuracy with $k_\true=4$, $s=3$, and varying settings of $N$ and $T$.
    The accuracy is similar using $k = 4$ (top) versus $k = 10$ (bottom).
    In both plots, the accuracy is averaged over $250$ independent trials.}
    \label{fig:accuracy}
\end{figure}

As the next result,
\fig\ref{fig:accuracy} shows the trajectory length $T$ is a major factor determining the classification accuracy \eqref{eq:acc}.
The accuracy reaches a threshold for each $T$ value as $N \rightarrow \infty$, and increasing $T$ exponentially increases the threshold.
These results in line with the theoretical analysis in Sec.~\ref{sec:analysis}, which also suggests an exponential dependence on $T$.

\begin{figure}[t]
    \centering
\includegraphics[width=0.9\linewidth]{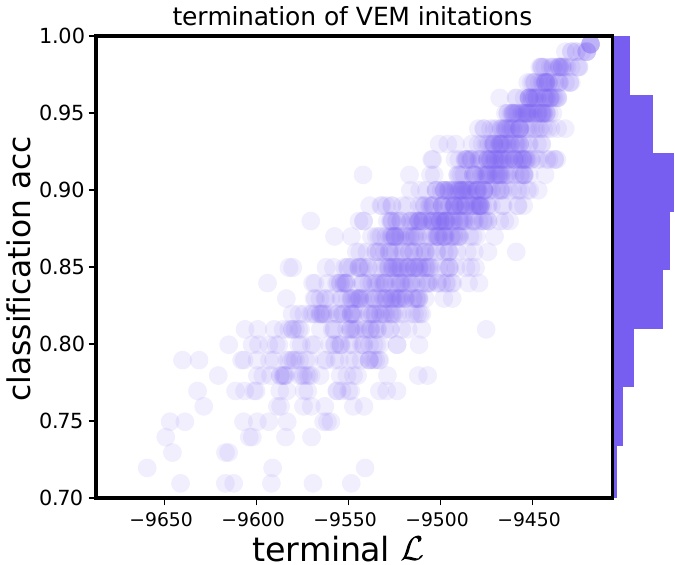}
    \caption{Demonstration of the local optima over $1000$ random variational EM initializations when $k = 15$, $k_\true = 10$, $s=7$, $N=100$, and $T=50$. 
    Many variational EM runs end with nearly optimal assignments, some do not.}
    \label{fig:EMlocalmin}
\end{figure}

EM and variational EM are known to terminate in locally optimal parameter values that can be far from the global optimum \cite{redner1984mixture}.
To safeguard against misconvergence, this paper randomly initializes variational EM many times and accepts the results that maximize the likelihood bound $\mathcal{L}$.
The outcomes of $10^3$ random initialization are shown in
\fig\ref{fig:EMlocalmin}.
For the figure, the parameters $k = 15$, $k_\true = 10$, $s=7$, $N=100$, and $T=50$ were deliberately chosen to make the optimization challenging and produce a range of variational EM solutions.
Perfect accuracy is achieved in some variational EM runs, but other runs result in low accuracy and low $\mathcal{L}$ values.
While the sensitivity of variational EM to the initial parameters may appear prohibitive, given the speed of the algorithm and its ease of parallelization, running hundreds to thousands of independent initializations is feasible and sufficient to yield reliable results.
Nonetheless, there remains significant future interest in identifying more principled initializations for variational EM or more robust stochastic variants \cite{celeux1996stochastic}. 

\subsection{\texttt{Last.fm} User Data}
\label{sec:lastfm}

The next experiment uses a publicly available data set that records the listening histories of \texttt{Last.fm} users in 2007.
\edit{The works \cite{kausik2023learning,GKV16} report trajectory-level clustering results on this \texttt{Last.fm} Markov-chain-mixture benchmark, making them natural baselines.}

Following the previous work \cite{kausik2023learning}, the ground-truth model is specified as follows.
The data contains $k_\true = 10$ Markov chains, corresponding to the \texttt{Last.fm} users with the greatest number of song listens.
The data is coded using $s=100$ states based on assigning a discrete label to each song in a user's listening history by the predominant genre.
If a user sequentially listens to songs from the same genre, these songs are collapsed into a single state. 
Hence, each Markov chain models the \emph{genre transitions} for a \texttt{Last.fm} user.
To make the learning problem difficult, each user's listening history is broken into $75$ equal-length segments, and the segments are truncated to a short trajectory length $20 \leq T \leq 100$.

\begin{figure}[t]
    \centering
    \includegraphics[width=0.85\linewidth]{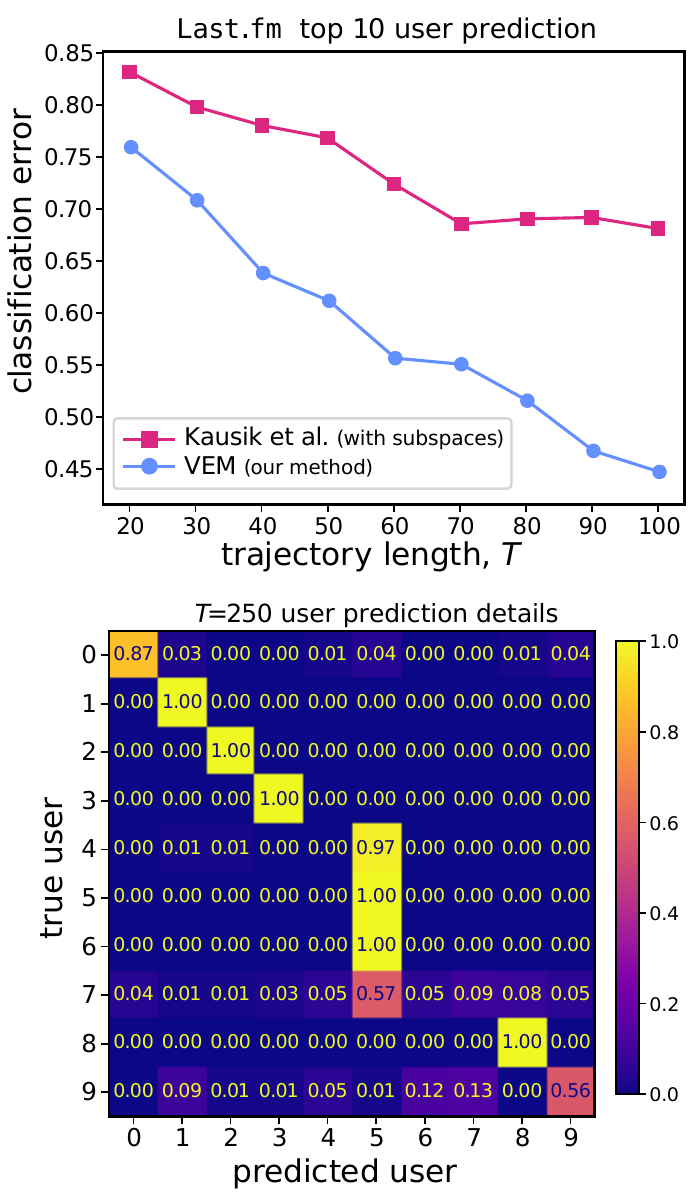}
    \caption{Classification accuracy for the top 10 \texttt{Last.fm} users, given $N=75$ trajectories and $s=100$ possible genres for each user.
    Top: Variational EM is twice as accurate as the best method reported in \cite{kausik2023learning}.
    Bottom: Confusion matrix for identifying the top 10 users with $T=250$.}
    \label{fig:lastfm1}
\end{figure}  

\edit{\fig\ref{fig:lastfm1} shows that the Bayesian variational EM approach with $k = k_\true = 10$ components achieves classification accuracy approximately $2\times$ higher than the experiments in \cite{kausik2023learning}.
The error is also much lower than the $90\%$ classification error in \cite{GKV16}.
These results are included for context, although the earliest benchmark \cite{GKV16} used more numerous and shorter trajectories ($N = 30,000$, $T = 3$), so the numbers are not directly comparable.}

While variational EM outperforms previous approaches, the misclassification rate for length $T=100$ trajectories is still notably high, $45\%$. 
To further improve the accuracy,
the trajectories were extended to length $T = 250$, resulting in a reduction to $30\%$ error.
The right panel of \fig\ref{fig:lastfm1}
shows the confusion matrix, which indicates that most of the misclassifications come from 4 users who all had similar listening histories.

\begin{figure}[t]
    \centering
    \includegraphics[width=\linewidth]{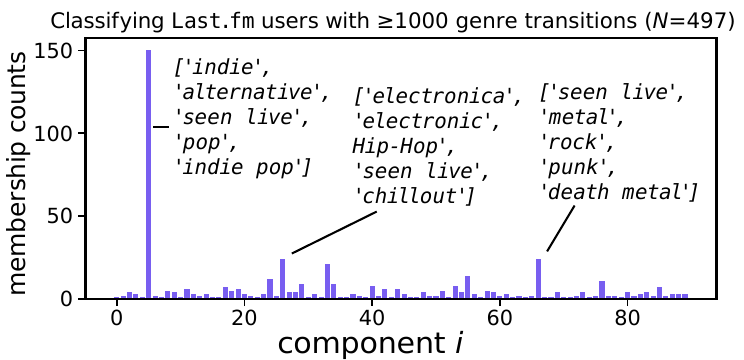}
    \caption{Classification of users for a \texttt{Last.fm} data set with $N = 497$ users, $s = 100$ genres, and trajectory length $T = 1000$. Popular components are labeled with the most frequent genres.}
    \label{fig:lastfm2}
\end{figure}

For further insight, a new data set was prepared with a length $T = 1000$ trajectory for each of the top $N = 497$ \texttt{Last.fm} users.
Variational EM was applied to this data set with a maximum number of $k=100$ components.
\fig\ref{fig:lastfm2} shows the mixture modeling results, which indicate that a few components have substantially large memberships.
The 3 largest components correspond to well-known genres of `indie', `electronic', and `metal'.
The overwhelming popularity of the indie component explains why distinguishing \texttt{Last.fm} users may be difficult.
Many indie rock listeners had quite similar listening histories in 2007 \cite{eck2007automatic}.

\subsection{Ultrarunners data set}
\label{sec:ultrarunners}

The third experiment is based on data from the 2012 International Association of Ultrarunners (IAU) World Championship held in Katowice, Poland \cite{bartolucci2015finite,roick2020clustering}.
The goal of this experiment is to find the predominant pacing patterns in the data and infer which pacing patterns lead to the best overall performance.
Conventional wisdom among ultrarunners contends that a slower pace to start is ideal \cite{berger2024limits}, as many runners tend to overexert themselves during the start and tire later during a race.

Pacing is an individualized measure of performance, so each runner's data over the 24-hour event was normalized based on average speed.
First, $12$ runners that did not complete a single lap were removed from the data set, yielding $N = 248$ total trajectories.
Then, each runner's average speed was calculated by dividing the total number of laps over the 2.314 km course by the total number of hours until finishing or dropping out of the event (typically 24 hours).
Each runner's data was normalized by dividing the number of laps per hour by the average speed.
Last, the data was categorized into $s = 4$ Markov states: ``resting'' when the hourly speed is less than $0.5$ times the average speed, 
``normal'' when the hourly speed is between $0.5$ and $1.5$ times the average speed, 
``strained'' when the hourly speed is greater than $1.5$ times the average speed,
and a terminal ``ended'' state.

\begin{figure}[t]
    \centering
    \includegraphics[width=\linewidth]{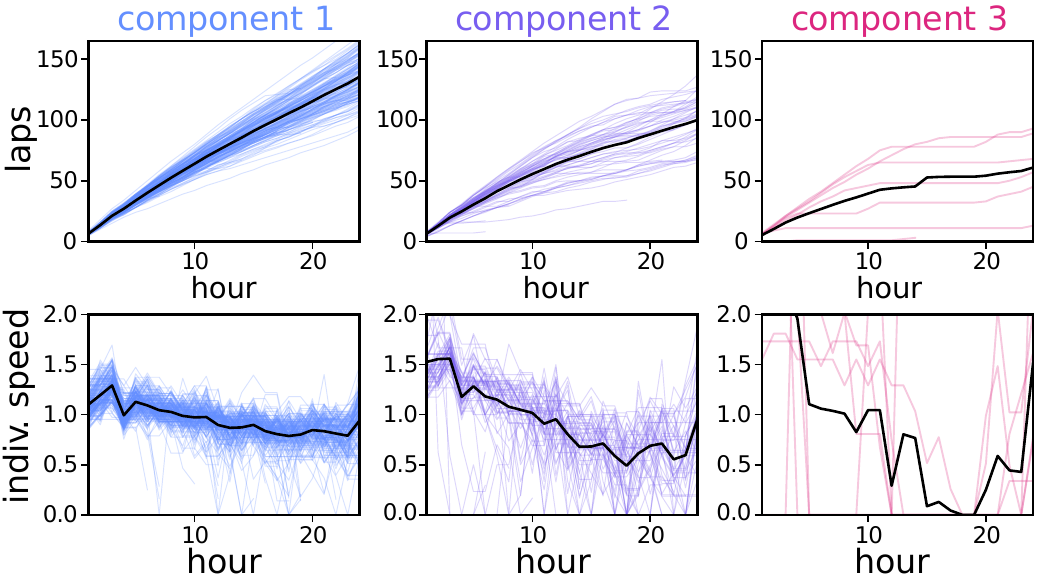}
    \caption{Application of Markov chain mixture modeling to the 2012 International Association of Ultrarunners (IAU) World Championship data.
    With a maximum of $k=10$ components, three running patterns emerged: (1) constant pace, (2) a fast start followed by subsequent resting, (3) erratic performance. 
    Black lines show component means.}
    \label{fig:ultra}
\end{figure}

The variational EM algorithm identifies three distinct running patterns in the data, as displayed in \fig~\ref{fig:ultra}.
The predominant pattern ($70\%$ of runners) involves a roughly constant pace, with slightly higher relative speed at the start of the race. 
The second-most-common pattern ($25\%$ of runners) involves overexertion at the start of the race and a significant slowdown later. 
The third pattern involves erratic running with no apparent strategy.
A previous analysis identified similar running patterns in the 2012 ultrarunner data \cite{bartolucci2015finite}.

The Markov state mixture model is interesting in two ways.
First, the runners in group 1 ran notably faster than those in group 2 even though the average speed was not included as input into the algorithm.
Rather, a constant pace emerged organically as the top-performing pattern.
Second, the ideal strategy of a slower opening pace was not observed in any group.
While the slow start is considered optimal, it is also notably difficult to achieve in races \cite{berger2024limits}. 

\subsection{Synthetic Gene Expression Data}
\label{sec:gene}

\begin{figure*}
    \centering
    \includegraphics[width=0.8\textwidth]{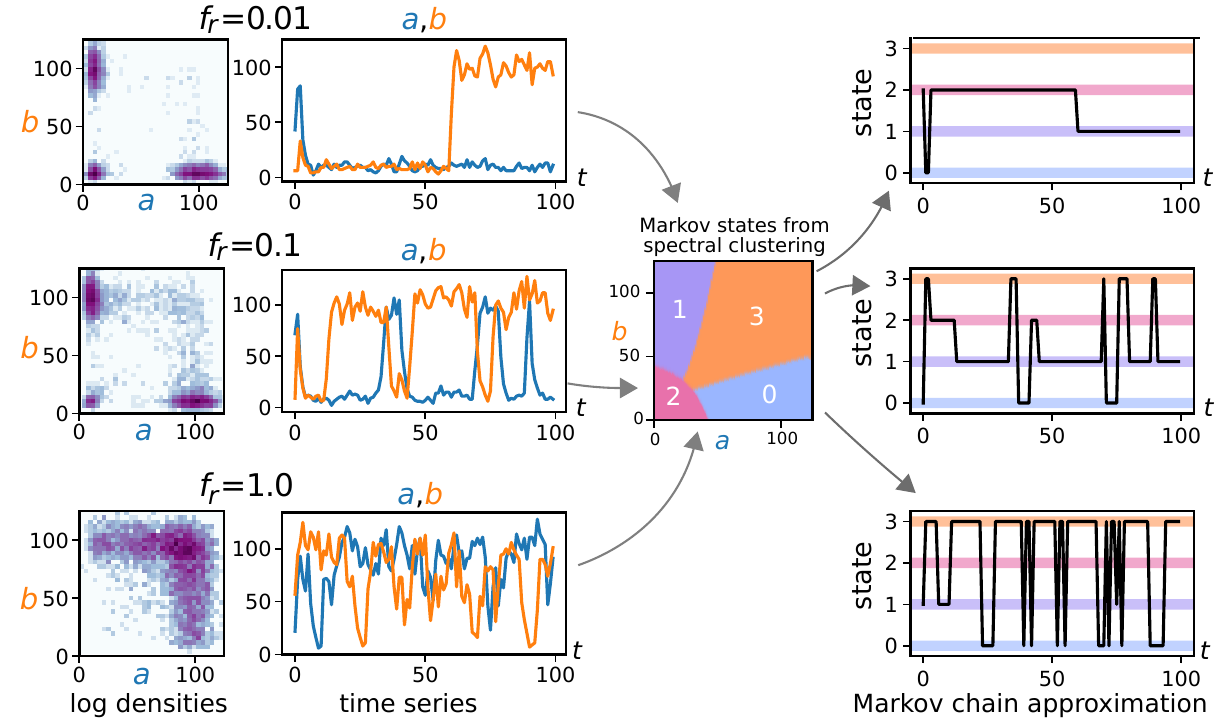}
    \caption{Schematic of finite state approximation for the mutual-inhibition, self-activation (MISA) gene circuit. 
    Counts of proteins $a$ and $b$ are visualized as scatter plots (column 1) and time series (column 2).
    $(a, b)$ trajectories are classified into 4 states (column 3),
    and trajectories are projected based on the 4 states (column 4).
    Rows indicate different magnitudes of the rate parameter $f_r = 0.01, 0.1, 1.0$.}
    \label{fig:MISA-setup}
\end{figure*}

The last experiment is based on synthetic data generated from the mutual-inhibition, self-activation (MISA) gene circuit. 
The MISA circuit has been observed in many biological systems \cite{graf2009forcing,huang2009reprogramming,zhou2011understanding, smith2016regulation}, and it has been studied extensively by theorists \cite{schultz2008extinction,morelli2008reaction,feng2012new,gallivan2020analysis}.
It consists of an $A$ gene and a $B$ gene that inhibit each other and activate themselves through the production of $a$ and $b$ proteins.
Mathematically, it can be modeled as a chemical reaction network \cite{anderson2015stochastic} containing the following species and reactions.
There is one $A$ gene and one $B$ gene that exhibit conditions $ij \in \{00, 01, 10, 11\}$, where $i = 1$ indicates the presence of an activator and $j = 1$ indicates the presence of a repressor.
The A and B genes produce proteins at a rate $g_{ij}$:
\begin{align*}
    A_{ij} &\overset{g_{ij}}{\to} A_{ij} + a, \quad ij={00,01,10,11}, \\
    B_{ij} &\overset{g_{ij}}{\to} B_{ij} + b, \quad ij={00,01,10,11},
\end{align*}
where $g_{00} = g_{01} = g_{11} = 10$ and $g_{10} = 100$.
The proteins degrade at rate $d = 1$:
\begin{equation*}
    a \overset{d}{\to} \varnothing, \quad
    b \overset{d}{\to} \varnothing.
\end{equation*}
Last, the proteins influence the conditions of the $A$ gene and $B$ gene as follows:
\begin{align*}
    A_{0 j} + 2a &\stackrel{h_a}{\underset{f_a}{\rightleftarrows}} A_{1 j}, \quad 
    B_{0 j} + 2b \stackrel{h_a}{\underset{f_a}{\rightleftarrows}} B_{1 j}, \quad
    j = 0,1, \\
    A_{i0} + 2b &\stackrel{h_r}{\underset{f_r}{\rightleftarrows}} A_{i 1}, \quad
    B_{i 0} + 2a \stackrel{h_r}{\underset{f_r}{\rightleftarrows}} B_{i 1}, \quad 
    i = 0, 1.
\end{align*}
Three of the rate parameters are fixed to $h_a=10^{-1}$, $f_a = 1$, and $h_r=10^{-3}$.
However, $f_r$ is a free parameter that controls the extent of protein activation versus repression in the system.
In summary, the MISA gene circuit can be encoded as a vector $X_t \in \{00, 01, 10, 11\}^2 \times \mathbb{N}^2$. The first two coordinates indicate the conditions of the $A$ gene and $B$ gene, and the last two coordinates indicate the populations of the $a$ and $b$ proteins.

The MISA trajectories were simulated using the stochastic simulation algorithm (SSA) implemented in \texttt{PyGillespie} \cite{matthew2023gillespy2}
with various values of the $f_r$ parameter.
The SSA algorithm produces a continuous-time trajectory $(X_t)_{t \geq 0}$; 
however, the states were sampled at uniformly spaced times $t = 0, 1, 2, \ldots$ to produce a discrete-time data set.

\fig\ref{fig:MISA-setup} displays the stochastic switching in the populations of $a$ and $b$ proteins.
The protein populations toggle between four metastable states, which are associated with each gene being ``on'' or ``off''. 
When the rate $f_r$ is low (top panels), the genes are frequently off, leading to protein populations of just 0--20.
When $f_r$ is high (bottom panels), both genes are frequently on, leading to increased populations of 50--150 proteins.
The most interesting behavior occurs for intermediate values of $f_r$ (middle panels), because then one gene is typically on and the other is off.
The genes actively compete to produce more proteins.

The MISA gene circuit is not Markovian in the protein numbers $a$ and $b$.
Nonetheless, a Markov state model can be constructed by applying spectral clustering \cite{coifman2005geometric} to the $(a, b)$ data.
The middle panel of \fig\ref{fig:MISA-setup} shows the outcome of spectral clustering using a Gaussian kernel with bandwidth $\sigma = 50$.
The four identified clusters are interpretable and they efficiently represent the metastable dynamics.

\begin{figure}[t]
    \centering
    \includegraphics[width=1\linewidth]{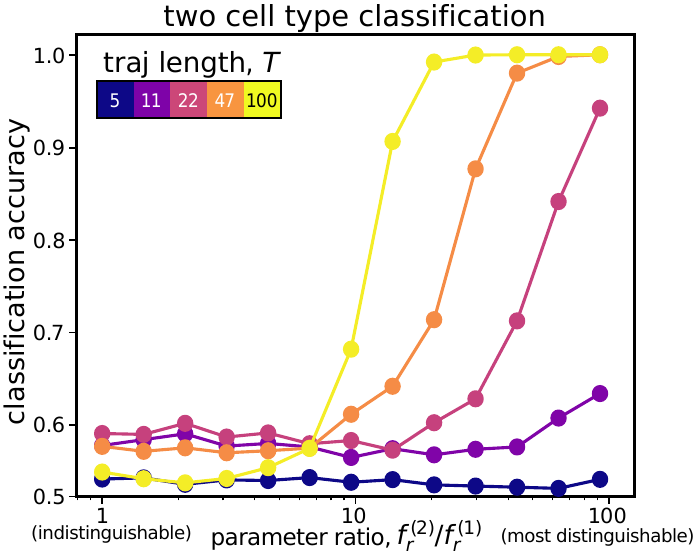}
    \caption{Classification accuracy for synthetic MISA data containing 15 trajectories with  $f_r^{(1)}=0.01$ and 15 trajectories with $f_r^{(2)} \neq f_r^{(1)}$. 
    Trajectories are reliably identified when $T \geq 25$ and $f_r^{(2)} / f_r^{(1)} \geq 25$.}
    \label{fig:fig_MISA_acc}
\end{figure}

Variational EM can be used to untangle multiple Markov chains from a single MISA data set.
This application area is increasingly important, as MISA models with variable parameters are now being engineered in labs \cite{li2018engineering} and applied to understand long chains of cell differentiation \cite{zhou2011understanding}.
As an example application, researchers have obtained \edit{time-series} data for a MISA gene circuit in the bacteriophage $\lambda$ switch \cite{fang2018cell}.
They analyzed the time-series data assuming a single Markov chain, but the Markov chain mixture model provides the ability to discern heterogeneities within the samples. 

To test the feasibility of the mixture modeling approach,
spectral clustering and variational EM were applied to a range of synthetic MISA data sets containing $15$ trajectories with an unbinding rate $f_r^{(1)} = 0.01$ and $15$ trajectories with a different unbinding rate $f_r^{(2)} \neq f_r^{(1)}$.
The algorithms were used to untangle the two populations of trajectories, leading to the results in
\fig\ref{fig:fig_MISA_acc}.
The results indicate that the classification accuracy depends greatly on the trajectory length $T$ and the parameter ratio $f_r^{(2)}/f_r^{(1)}$.
When $f_r^{(2)}/f_r^{(1)} \geq 25$ and $T \geq 25$, the populations can be separated with near-perfect accuracy.
However, short trajectories with $T \leq 5$ cannot be reliably separated for any $f_r^{(2)}$ parameter.
These results support the theoretical analysis in Thm.~\ref{thm:1}, since the Kullback-Leibler divergence increases with $f_r^{(2)}/f_r^{(1)}$ and increases linearly with $T$.

\section{Conclusion} \label{sec:conclude}

This paper has proposed an extension of Markov state modeling that enables the study of heterogeneities within \edit{time-series} data.
Previous work employed Markov state modeling under the assumption that the observations come from a homogeneous population modeled by a single Markov chain.
This paper develops the theory and practice of learning a mixture of different Markov chains simultaneously.

The paper has motivated and tested a variational EM algorithm that automatically identifies the number of chains and the dynamics of each chain. 
\edit{The proposed mixture modeling approach deliberately combines classical, well-understood components, based on the discretization step from Markov state modeling and the fitting step from variational EM.
It contributes an interpretable and competitive baseline for heterogeneous dynamical data.}
The algorithm is computationally efficient: unlike past work, it identifies the number of chains organically, without relying on expensive model comparisons or posterior sampling.

The current bottleneck, shared among all the standard methods for fitting Markov chain mixture models (EM and variational EM), is the convergence to locally optimal parameters.
Running these methods with many random initializations is the currently proposed solution, but it is the computationally limiting aspect of the approach.
Future work should identify and rigorously justify an initialization strategy that works well for variational EM with any trajectory length $T$.

\edit{The paper also provides a theorem that lower bounds the classification error for any Markov chain mixture model, similar to classical information-theoretic bounds.}
The bound is stated in terms of the Kullback-Leibler divergence between the underlying Markov chains, and it \edit{suggests the optimal classification error decreases exponentially in $T$ because the trajectory-level KL divergence grows linearly in $T$}.
This prediction is repeatedly supported in the numerical experiments (Secs.~\ref{sec:synthetic}, \ref{sec:lastfm}, \ref{sec:gene}).
The theory precisely quantifies the established wisdom \cite{ramoni2002bayesian} that long trajectories are better than short ones.
\edit{The bound requires no ergodicity or mixing-time assumptions, making it applicable to transient or non-stationary dynamics common in short experimental trajectories.}

Last, there is a natural question: when should scientists use Markov chain mixture models in the future?
Three of the four experiments (Secs.~\ref{sec:lastfm}, \ref{sec:ultrarunners}, \ref{sec:gene}) involved real data or real biological systems, and the mixture modeling approach led to accurate and interpretable results.
One of the authors (C.E.M.) listened to indie rock on \texttt{Last.fm} in 2007 while the other author (R.J.W.) runs ultramarathons competitively, and they can attest to the qualitative accuracy of the Markov chain mixture models. %
\edit{More broadly, these results reinforce a simple point: careful probabilistic modeling and inference, even when built from classical ingredients, can remain competitive in modern high-throughput settings, serving as an interpretable baseline against which newer approaches are compared.}

In conclusion, Markov chain mixture models have a track record of distinguishing meaningful heterogeneities in human behavior, for example, heterogeneous patterns of surfing the internet \cite{melnykov2016clickclust}, commuting between home, work, and school \cite{zhou2021you}, listening to music (Sec.~\ref{sec:lastfm}), or running ultramarathons (Sec.~\ref{sec:ultrarunners}).
The variational EM algorithm is \edit{an efficient strategy} for fitting these models while automatically selecting the number of components.
Moving forward, Markov chain mixture models could be extended to other areas (chemistry, biology, climate science) where Markov state modeling is already popular.
Indeed, the experiments in Sec.~\ref{sec:gene} demonstrate promising first steps toward applications in the rapidly evolving area of gene expression data analysis \cite{eisen1998cluster,ernst2005clustering,mcdowell2018clustering,mitra2021rvagene}.

\bmhead{Acknowledgements}

The authors thank Elizabeth Read and Adam MacLean for helpful discussions about gene expression time series.
C.E.M. was partially supported by a University of California Society of Hellman Fellows fund and NSF CAREER DMS-2339241. 
R.J.W. was supported by the Office of Naval Research through BRC Award N00014-18-1-2363, the National Science Foundation through FRG
Award 1952777, and Caltech through the Carver Mead New Adventures Fund, under the aegis of Joel A. Tropp.

\bmhead{Code Availability}

Code to reproduce the numerical experiments can be found in the \edit{GitHub} repository \url{https://github.com/chris-miles/MarkovChain-VEM}.

\bibliography{references}%

\end{document}